\documentclass[12pt]{amsart}
\usepackage[notquote]{hanging}
\usepackage{setspace}
\usepackage{lscape}
\usepackage{amsthm}
\usepackage{amsmath}
\usepackage{amssymb}
\usepackage{amsfonts}
\usepackage{subfig}
\usepackage{graphicx}
\usepackage{yfonts}
\usepackage{fullpage}
\usepackage{color}
\theoremstyle{definition}
\newtheorem{mydef}{Definition}
\newtheorem{lemma}{Lemma}
\newtheorem{theorem}{Theorem}

\newtheorem{prop}{Proposition}

\newcommand{\PP}{{\mathbb P}}
\newcommand{\EE}{{\mathbb E}}

\newcommand{\Keywords}[1]{\par\noindent 
{\small{\em Keywords\/}: #1}}

\author{Mareike Fischer, Michelle Galla, Lina Herbst and  Mike Steel}

\address{Allan Wilson Centre \\
University of Canterbury \\
Christchurch, New Zealand}

\address{Department for Mathematics and Computer Science, Ernst-Moritz-Arndt University \\ Greifswald, Germany}
\title{The most parsimonious tree for random data}
\date{\today}

\begin{document}

\maketitle

\doublespacing
\begin{abstract}
\noindent
Applying a method to reconstruct a phylogenetic tree from random data provides a way to detect whether that method has an inherent bias towards certain tree `shapes'.  For maximum parsimony, applied to a sequence of random 2-state data, each possible binary phylogenetic tree has exactly the same distribution for its parsimony score.  Despite this pleasing and slightly surprising symmetry, some binary phylogenetic trees are more likely than others to be a most parsimonious (MP) tree for a sequence of $k$ such characters, as we show.   For $k=2$, and unrooted binary trees on six taxa, any tree with a caterpillar shape has a higher chance of being an MP tree than any tree with a symmetric shape.  On the other hand, if we take any two binary trees, on any number of taxa,  we prove that this bias between the two trees vanishes as the number of characters grows. However, again there is a twist:  MP trees on six taxa are more likely to have certain shapes than a uniform distribution on binary phylogenetic trees predicts, and this difference does not appear to dissipate as $k$ grows. \\

\end{abstract}

\section{Introduction}

The `shape' of reconstructed evolutionary trees is of interest to evolutionary biologists, as it should provide some insight into the processes of speciation and extinction 
\cite{al1, al2, hey, hol, lam, sat}.  In this paper, `shape' refers just to the  discrete shape of the tree (i.e. we ignore the branch lengths); the advantages of this are 
that it simplifies the analysis, and it also confers a certain robustness (i.e. the resulting probability distribution on discrete shapes is often independent of the fine details
of an underlying speciation/extinction model \cite{al0}, \cite{lam}).  For example, if all speciation (and extinction) events affect all taxa at any given epoch in the same way, then we should expect the shape of  a reconstructed tree to be that predicted by the discrete `Yule--Harding'
model \cite{al1, har, lam}.  In fact, a general trend (see e.g. \cite{al1}) is that the shape of phylogenetic trees  reconstructed from biological data tends to be a little less balanced than this model predicts, but is more balanced than what would be obtained under a uniform model in which each binary phylogenetic tree has the same probability (this model is sometimes also called the `Proportional-to-Distinguishable-Arrangements' (PDA) model).   

There are, however, other factors which can lead to biases in tree shape. One is non-random sampling of the taxa on which to construct a tree (influenced, for example, by the particular interests of the biologists or the application of a certain strategy to sample taxa).  Another cause of possible bias is  that a tree reconstruction method may itself  have an inherent preference towards certain tree shapes. A way to test this latter possibility is to apply the tree reconstruction method to data that contain no phylogenetic signal at all, in particular, purely random data, where each character is generated independently by a process that assigns states to the taxa uniformly (e.g.  by  the toss of a fair coin in the case of two states).  For some methods, such as `TreePuzzle', such data leads to very balanced trees (similar to the Yule--Harding model \cite{treepuzzle, zhu}).
However, other methods, such as maximum likelihood and maximum parsimony, lead to less balanced trees, that are closer in shape to the uniform model, as recently reported in \cite{hol}.  In the case of maximum parsimony, the two-state symmetric model has the even-handed  property that every binary tree has exactly the same distribution of its parsimony score on $k$ randomly generated characters.  Thus, it might be supposed that the maximum parsimony (MP)  tree for such a sequence of characters would also follow a uniform distribution. However, while this holds in special cases, it does not hold in general, as we show below.

\subsection{Trees and parsimony: definitions and basic properties}

In phylogenetics, graphs, especially trees, are used to describe the ancestral relationships among different species. A main goal of phylogenetics is to infer an evolutionary tree from data available from present-day species.
In graph theory, a tree $T=(V,E)$ consists of a connected graph with no cycles. Certain leaf-labelled trees (`phylogenetic trees') are widely used where the set of extant species label the leaves and the remaining vertices represent ancestral speciation events \cite{felsenstein}.
There are different methods of reconstructing a phylogenetic tree. One of the most famous tree reconstruction methods is maximum parsimony. For a given tree and discrete character data, the parsimony score can be found in polynomial time by using the Fitch--Hartigan algorithm \cite{fit, hart}. The parsimony score counts the number of changes (mutations) required on the tree to describe the data. 
This problem of finding the optimal parsimony score for a given tree is often called the `small parsimony' problem. The `big parsimony' problem aims at finding the most parsimonious tree (`MP tree') amongst all possible trees. This problem has been proven to be NP-hard \cite{np}. 

In this paper, we assume that each taxon from the leaf set $X$ of the tree is assigned a binary state (0 or 1) independently, and with equal probability.  This process is then repeated (also independently) to generate a sequence of characters (defined formally below).   For binary trees with random data, we are interested in the probability that a tree is an MP                tree, and also what happens when the length of the sequences or the number of leaves gets larger. 
In particular, we wish to determine whether each tree is equally likely to be selected as an MP  tree.  

\begin{mydef}{[Binary phylogenetic trees]} \label{xtree}
An \textit{(unrooted) binary phylogenetic X-tree} is a tree $T$ with leaf set $X$ and with every interior (i.e. non-leaf) vertex of degree  exactly three. We will let $UB(X)$ be the set of unrooted binary phylogenetic $X$-trees. When $X = [n] = \{1,\dots,n\}$, we will write $UB(n)$. 
\end{mydef}
\begin{mydef}{[Character, extension, parsimony score]}
\begin{itemize}
\item
A \textit{character on $X$} over a finite set $R$ of character states is any function $f$ from $X$ into $R$; $
f: X \rightarrow R.$
In this paper we will consider two-state characters;
$f: X \rightarrow \{ 0,1 \}.$

\item
A function $\bar{f}: V \rightarrow R$ such that $\bar{f}|_X =f$ is said to be an \textit{extension} of $f$ since it describes an assignment of states to all vertices of $T$ that agrees with the states that $f$ stipulates at the leaves.

\item
Let $ch(\bar{f},T) := \left|\{ e=\{u,v\} \in E: \bar{f}(u) \neq \bar{f}(v) \}\right|$ be the \textit{changing number} of $\bar{f}$. Given a character $f: X \rightarrow R$, the \textit{parsimony score of $f$} on $T$, denoted $ps(f,T)$, is the smallest changing number of any extension of $f$, i.e. :
\begin{align*}
ps(f,T) := \min_{\bar{f}:V \rightarrow R,\bar{f} |_X = f} \{ ch(\bar{f},T) \}.
\end{align*}
An extension $\bar{f}$ of $f$  for which $ch(\bar{f},T) = ps(f,T)$ is said to be a \textit{minimal extension}. \\
Let $\mathcal{C} = (f_1,\dots,f_k)$ be a sequence of characters on $X$. The \textit{parsimony score of $\mathcal{C}$} on $T$, denoted $ps(\mathcal{C},T)$, is defined by
$
ps(\mathcal{C},T) := \sum_{i=1}^{k} ps(f_i,T).
$

\end{itemize}

\end{mydef}

\section{Comparing given trees}
Let $X_k(T)$ be the parsimony score of $k$ random two-state characters on $T \in UB(n)$. We will see shortly (Proposition~\ref{unif}) that the distribution of $X_k(T)$ does not depend on the shape of $T$; it just depends on $n$.  Notice that $X_k(T) = X_1 + X_2 + \dots + X_k,$ where $X_i$ (for $i=1, \dots, k)$ form a sequence of independent and identically distributed random variables (with common distribution $X_1(T)$). 
If $\PP(X_k(T) = l)$ denotes the probability that $T$ has parsimony score $l$  then, from \cite{psp}, we have, for each $l \in [1, \lfloor n/2 \rfloor]$: 
\begin{align}
\PP(X_1(T) = l) = \frac{2 n - 3 l}{l} \cdot \binom{n-l-1}{l-1} \cdot 2^{l-n}, \label{eq:formel} \end{align}  
with $\PP(X_1(T) =0) = 2^{1-n}$ and $\PP(X_1(T)=l)=0$ for $l> \lfloor n/2 \rfloor$.
Furthermore, 
$\EE[X_1(T)] = \frac{3 n - 2 - (- \frac{1}{2})^{n-1}}{9} \sim \frac{n}{3}$
is the expected parsimony score of $T$, and $\EE[X_k(T)] = k \cdot \EE[X_1(T)].$  
An immediate consequence of (\ref{eq:formel}) is the following. 

\begin{prop}
\label{unif}
For every $k \geq 1$ and $n \geq 2$, the distribution of the parsimony score of $k$ independent random binary characters (i.e. $X_k(T)$) is the same for all $T \in UB(n)$. 
\end{prop}

\bigskip

\subsection{Comparing two trees by their parsimony score}
We begin this section by describing a tree rearrangement operation on binary phylogenetic trees \cite[Chapter 2.6]{sem}, namely tree bisection and reconnection (TBR).
 Let $T$ be a binary phylogenetic $X$-tree and let  $e = \{ u,v \}$ be an edge of $T$.
A TBR operation is described as follows. Let $T^{\prime}$ be the binary tree obtained from $T$ by deleting $e$, adding an edge between a vertex that subdivides an edge of one component of $T \setminus e$ and a vertex that subdivides an edge of the other component of $T \setminus e$, and then suppressing any resulting degree-two vertices. In the case that a component of $T \setminus e$ consists of a single vertex, then the added edge is attached to this vertex. $T^{\prime}$ is said to be obtained from $T$ by a single TBR operation. 
\begin{prop}
\label{bias}
Let $T, T^{\prime} \in UB(n)$.
\begin{itemize}
\item[(i)] If $T$ and $T^{\prime}$ are one TBR apart, then $\PP(X_k(T) < X_k(T^{\prime})) = \PP(X_k(T^{\prime}) < X_k(T))$ holds for all $k \geq 1$. 
\item[(ii)] If $T$ and $T^{\prime}$ are more than one TBR apart, then the equality  $\PP(X_k(T) < X_k(T^{\prime})) = \PP(X_k(T^{\prime}) < X_k(T))$ can fail, even for $k=1$ and $n=6$.\end{itemize}
\end{prop}
\begin{proof}
\mbox{ }
\begin{itemize}
\item[(i)]  
From \cite[Lemma 5.1]{bryant}, if $T$ and $T^{\prime}$ are one TBR apart then for any character $f$,
$\left| ps(f,T) - ps(f,T^{\prime}) \right| \leq 1.$
In particular,
\begin{equation}
\label{lip}
\left| X_1(T) - X_1(T^{\prime}) \right| \leq 1.
\end{equation}
For $k \geq 1$, let $\Delta_k = X_k(T)-X_k(T')$. Then 
if $T, T^{\prime} \in UB(n)$ are one TBR apart, then $\Delta_1 = X_1(T) - X_1(T^{\prime})$ is either $0, 1$ or $-1$, by (\ref{lip}).  
Moreover,  $\PP(\Delta_1= m) = \PP(\Delta_1 = -m)$ for all $m \in \{0,1-1\}$, since
$\EE[\Delta_1] = 0$, by Proposition~\ref{unif}.
Furthermore,  $\Delta_k = D_1 + \cdots + D_k$, where $D_1, \ldots, D_k$ are independent and identically distributed as $\Delta_1$, so we have: 
\begin{align*}
\PP(\Delta_k = m) &= \sum_{\substack{m_1,\dots, m_k  \in \{-1,0,1\}: \\ m_1+ \dots + m_k=m}} \PP(D_1=m_1 \wedge D_2=m_2 \wedge \cdots \wedge D_k=m_k) \\
&= \sum_{\substack{m_1,\dots, m_k \in \{-1,0,1\}: \\ m_1+ \dots + m_k=m}} \prod_{j=1}^{k} \PP(D_j=m_j) 
= \sum_{\substack{m_1,\dots, m_k  \in \{-1,0,1\}: \\ m_1+ \dots + m_k=m}} \prod_{j=1}^{k} \PP(D_j=-m_j) \\
&= \sum_{\substack{m_1^{\prime},\dots, m_k^{\prime}  \in \{-1,0,1\}: \\ m_1^{\prime}+ \dots + m_k^{\prime}=-m}} \PP(D_1=m_1^{\prime} \wedge D_2=m_2^{\prime} \wedge \cdots \wedge D_k=m_k^{\prime}) 
= \PP(\Delta_k = -m).
\end{align*}
This provides the equality $\PP(X_k(T) < X_k(T^{\prime})) = \PP(X_k(T^{\prime}) < X_k(T))$ for all $k \geq 1 $. 
\item[(ii)] We prove this by exhibiting one counterexample, namely the trees shown in Fig.~\ref{diff}.  Let $\Delta_k = X_k(T) - X_k(T^{\prime})$. The equality $\PP(X_k(T) < X_k(T^{\prime})) = \PP(X_k(T^{\prime}) < X_k(T))$ is equivalent to $\PP(\Delta_k < 0) = \PP(\Delta_k > 0)$.

\begin{figure}[ht]
\centering
\includegraphics[scale=1.0]{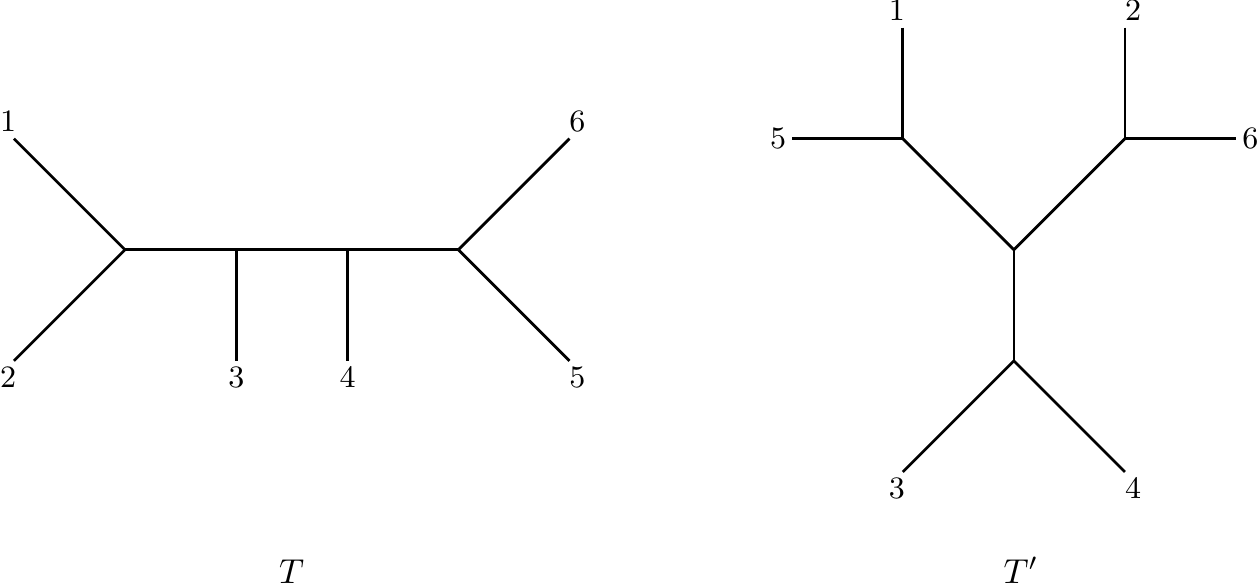}
\caption{Two trees $T, T^{\prime}  \in UB(6)$. Note that $T$ and $T^{\prime}$ are more than one TBR apart.}
\label{diff}
\end{figure}

\noindent
By calculating the parsimony score for the 32 different two-state characters (without loss of generality  we set $f(1) := 0$) we can assign the values that $\Delta_1$ can take and the probability of those values. 
$\Delta_1 = -2$ occurs precisely when  $X_1(T) = 1$ and $X_1(T^{\prime}) = 3$ with probability $p = \frac{1}{32}$. 
$\Delta_1 = -1$ occurs precisely when  $X_1(T) = 1$ and $X_1(T^{\prime}) = 2$ or $X_1(T) = 2$ and $X_1(T^{\prime}) = 3$ with probability $q = \frac{3}{32}$. 
$\Delta_1 = +1$ occurs precisely when $X_1(T) = 2$ and $X_1(T^{\prime}) = 1$ or $X_1(T) = 3$ and $X_1(T^{\prime}) = 2$ with probability $r = \frac{5}{32}$. 
Since $\Delta_1 = +2$ is not possible, $\Delta_1 = 0$ with probability $1 - (p+q+r) = \frac{23}{32}$. 
This leads to $\PP(\Delta_1 < 0) = \PP(\Delta_1 = -2) + \PP(\Delta_1 = -1) = \frac{4}{32} < \frac{5}{32} = \PP(\Delta_1 = +1) = \PP(\Delta_1 > 0)$. Therefore $\PP(X_k(T) < X_k(T^{\prime})) < \PP(X_k(T^{\prime}) < X_k(T))$ holds for $k=1$ and the choice of $T$ and $T'$ shown in Fig.~\ref{diff}. In other words,  the probability that the  symmetric tree $T$ is more parsimonious than the caterpillar tree $T'$ (on a single random binary character)  is higher than the probability that $T'$  is more parsimonious than $T$.
\end{itemize}
\end{proof}

\section{Maximum parsimony trees}

\begin{mydef}{[Maximum parsimony tree]}
Given a sequence ${\mathcal C} = (f_1, \ldots, f_k)$ of characters on $X$, a phylogenetic tree $T$ on $X$ that minimizes $ps(\mathcal{C},T)$ is said to be a \textit{maximum parsimony (MP) tree} for $\mathcal{C}$. The  corresponding ps-value is the \textit{parsimony or MP score of $\mathcal{C}$}, denoted $ps(\mathcal{C})$.
\end{mydef}

\noindent
\textbf{Notation:} Given $T \in UB(n)$, let $mp_{k}(T)$ denote the probability that $T$ is an MP tree for $k \geq 1$ random two-state characters on $[n]$. That is
$$mp_k(T):= \PP(X_k(T) \leq \min_{T^{\prime} \in UB(n)} \{ X_k(T^{\prime})\} ).$$
Notice that $mp_k(T)$ is not a probability distribution on $UB(n)$ since the positive probability of ties for the most parsimonious tree ensures that the $mp_k(T)$ values will
sum to a value greater than 1.

\begin{lemma}
\label{shape}
If  $T_1, T_2 \in UB(n)$  have the same shape then $mp_k(T_1) = mp_k(T_2)$.
\end{lemma}
\begin{proof}
Let $k \geq 1$ and  let $f_1, \dots, f_k$ be two-state characters. Then $ps(f_1, \dots, f_k, T) = ps(f_1^{\sigma}, \dots, f_k^{\sigma}, T^{\sigma})$,  where $\sigma$ is an element of the group $S_n$ of permutations on the leaf set $[n]$ of $T$. Notice that the map 
$f = (f_1, \dots, f_k) \mapsto f^{\sigma} = (f_1^{\sigma}, \dots, f_k^{\sigma})$ is a bijection, so the number of characters $f$ for which $T$ is an MP tree for $f$ equals the number of  characters $f$ for which $T^{\sigma}$ is an MP tree for $f$.
\end{proof}

It follows from Proposition~\ref{unif} and Lemma~\ref{shape} that if $n \geq 3$ and $k=1$, or if $k \geq 1$ and $n\leq 5$, then $mp_{k}(T)$ is constant for all $T \in UB(n)$.  However, this does not hold more generally, as we now state.

\begin{theorem}
\label{surprise}
$mp_{k}(T)$ is not constant for all $T \in UB(n)$ when $n=6$ and $k=2$. In particular, any given caterpillar tree (like $T$ in Fig.~\ref{diff})  has a higher probability of being an MP tree 
than a  symmetric tree (like $T'$ in Fig.~\ref{diff}).  
 \end{theorem}
 
The proof of Theorem~\ref{surprise} requires a detailed case analysis to identify the MP tree(s) for all pairs of characters $(f_1, f_2)$; 
details are provided in the Appendix.  The result is also confirmed by simulations, which are provided in the following section. 

\section{Asymptotic analysis}
We first show that the bias exhibited in Proposition~\ref{bias}(ii) disappears asymptotically but the bias apparent in Theorem~\ref{surprise} does not.
\begin{prop}\label{central}
For all $T, T^{\prime} \in UB(n)$ and all $n$:
\begin{align*}
\lim_{k \rightarrow \infty} \PP(X_k(T) < X_k(T^{\prime})) = \frac{1}{2}.
\end{align*}
\end{prop}
\begin{proof}
Let $T, T^{\prime} \in UB(n)$ and $k \geq 1$, and let
$ \Delta_k = X_k(T) - X_k(T^{\prime})  = D_1 + D_2 + \dots + D_k,$
where the random variable $D_i = ps(f_i,T) - ps(f_i,T^{\prime}) ~ (i=1,\dots,k)$ and the $D_i$ are independent and identically distributed. Moreover $\EE[D_i] = 0$ and $D_i$ has a standard
deviation $\sigma$ that is strictly positive and finite. To see that $\sigma>0$, note that $\sigma^2\geq \PP[D_i \neq 0]$ by Chebychev's inequality, and  $D_i$ is nonzero whenever
$f_i$ corresponds to a two-state character that has parsimony score 1 on one of the trees $T, T^{\prime}$ and parsimony score greater than 1 on the other tree (at least one such  character must exist, since $T \neq T^{\prime}$, and every tree is uniquely determined by its characters of parsimony score 1). 
We can now apply the standard Central Limit Theorem to deduce that for an asymptotically standard normal variable $Z_k = \frac{\Delta_k - \EE[\Delta_k]}{\sigma \cdot \sqrt{k}}$,  we have:
\begin{align*}
\PP(\Delta_k < 0) &= \PP\left(Z_k < \frac{0-0}{\sigma \cdot \sqrt{k}}\right)  \overset{k \rightarrow \infty}{\longrightarrow} \frac{1}{2}.
\end{align*}
\end{proof}

Finally, we consider the limiting behaviour of $mp_k(T)$ as $k \rightarrow \infty$, and present simulations that suggest that even for $n=6$, this probability depends on the shape of the tree.  It is easily shown  that, for any $n>1$,  as $k \rightarrow \infty$, there is a unique most parsimonious tree, so $\sum_{T \in UB(n)}\lim_{k \rightarrow \infty} mp_k(T) =1$ (see e.g.  \cite{zhu} (Theorem 4(2)).  In other words,  $ \lim_{k \rightarrow \infty} mp_k(T)$ is a probability distribution on $UB(n)$. However, the additional claim there that $mp_k(T)$ is uniform on $UB(n)$  does not hold when $n=6$ and when either $k=2$ (Theorem~\ref{surprise}) or, it seems,  as $k\rightarrow \infty$, as we now explain. 

\subsection{Simulations} \label{simulations}
We used the computer algebra system {\em Mathematica} to generate alignments of lengths 2, 10, 100, 1,000, 10,000 and 100,000, respectively, by sampling characters for six taxa uniformly at random out of the 32 possible binary characters (we assume without loss of generality that the state of taxon 1 is fixed, say, to state $0$, whereas all other taxa can choose states $0$ or $1$). For each alignment, we ran an exhaustive search through the tree space of 105 unrooted binary phylogenetic trees in order to find all MP  trees. For each alignment length, we did 1,000 runs and we counted the average number of MP trees, as well as the number of times that each of the two tree shapes for six taxa (the caterpillar shape or the symmetric shape of $T$ and $T'$ in Fig.~\ref{diff}) were amongst the MP trees. We then calculated the ratio of the number of MP trees with a symmetric shape divided by the total number of MP trees. Note that this ratio should equal $\frac{1}{7}\approx 0.142857$ if both tree shapes were equally likely, because 15 out of the 105 possible  trees have the symmetric shape. However, the last column of Table 1 reveals  that only for the extremely short alignment of length 2 the ratio is close to this value
in our simulations (and it is not exactly equal to it, by Theorem~\ref{surprise}). Moreover,  the ratio  decreases away from $\frac{1}{7}$ as the alignment length increases (the small variation  at alignment length 10,000 is within one standard deviation).  This trend and the reported values strongly suggest that the limiting value of $mp_k(T)$ is not uniform across all trees in $UB(6)$. Note also that column 2  of Table 1 is also consistent with  the finding mentioned earlier that there will be a unique MP tree a with probability converging to $1$ as $k$ grows.

\begin{table}[htp]\footnotesize
 \begin{tabular}{cccccc}
 \hline
 Al. length & av. \# MP trees & \# symmetric tree was MP& \# caterpillar was MP & $\frac{ \text{\# symmetric MP trees}}{\text{\# MP trees}}$\\
 \hline
 2  & 17.177 & 2,375 & 14,802 & 0.138266\\
 10  & 3.908 & 365 & 3,543 & 0.0933982\\
 100  & 1.622 & 119 & 1503 & 0.0733662\\
 1,000  & 1.166 & 59 & 1107 & 0.0506003\\
 10,000  & 1.053 & 57 & 996 & 0.0541311\\
100,000  & 1.013 & 46 & 967 & 0.0454097\\
 \end{tabular}
 \caption{Overview of simulation results: For each alignment length, 1,000 runs were evaluated.}
  \end{table}

\subsection{Concluding comments}
In one sense, the two-state symmetric model is as favourable to all binary phylogenetic trees as is possible under maximum parsimony, since each tree has exactly the same probability distribution on the
parsimony score of $k$ random characters.  Moreover,  Proposition~\ref{central} shows that no one tree is any more likely to be an MP tree than another.  It may seem somewhat surprising, therefore, that the distribution of MP trees is not uniform, even asymptotically; however this has a simple explanation.  Although the characters are generated independently, and their parsimony scores is also independently distributed on any given binary tree, the MP binary tree is chosen once the  $k$ characters are given.  Thus these characters
are not independent random variables once we condition on a given tree being the MP tree for these characters.   Moreover, once one moves away from the simple two-state model (for example, to the $r$-state symmetric model) even the uniformity of MP scores on fixed trees disappears \cite{psp}.  In summary, while maximum parsimony on random data seems, in certain senses (described above),  to favour each binary tree equally, the method nevertheless exhibits a bias towards trees with certain tree shapes.

\subsection{Acknowledgments} We thank the Allan Wilson Centre for help funding this work. We also thank David Bryant for pointing out that MP trees might not
be uniformly distributed on sequences of random characters.

\section{Appendix: Proof of Theorem 1}

We first recall some definitions and establish some preliminary lemmas.
\begin{mydef}{[$X$-split, compatible]}
\label{co}

\noindent An \textit{$X$-split} is a bipartition of $X$ into two nonempty subsets, written $A \mid B$.
Given any phylogenetic $X$-tree $T$, if we delete any particular edge $e$ of $T$ and consider the leaf sets of the two connected components of the resulting disconnected graph we obtain an $X$-split, which is called a {\em split of} $T$ (corresponding to $e$). 
If two $X$-splits $A \mid B$ and $A^{\prime} \mid B^{\prime}$ of the some unrooted phylogenetic $X$-tree have the property that one of the four intersections $A \cap A^{\prime}$,  $A \cap B^{\prime}$, $B \cap A^{\prime}$, $B \cap B^{\prime}$ is empty, then $A \mid B$ and $A^{\prime} \mid B^{\prime}$ are said to be {\em compatible}.
A two-state character $f$ on $X$ is said to be \textit{compatible} with a phylogenetic $X$-tree $T$ if $f^{-1}(0) \mid f^{-1}(1)$ is an $X$-split of $T$.
Moreover,  a pair of two-state characters $f_1$ and $f_2$ on $X$ are said to be compatible with each other if $f_1$ and $f_2$ induce compatible $X$-splits (equivalently, if there exists a phylogenetic $X$-tree
that $f_1$ and $f_2$ are compatible with).
Finally, a two-state character $f$ on $X$ is {\em constant} if $f(x)$ takes the same value (0 or 1) for all $x \in X$.
\end{mydef}

The following result is easily established, with Part (b) following from the Split-Equivalence Theorem \cite[Theorem 3.1.4]{sem}.
\begin{lemma}
\label{tc}
\mbox{ }
\begin{itemize}
\item[(a)]  A two-state character $f$ is compatible with $T$ if and only if $ps(f,T)=1$.
\item[(b)]  A pair of two-state characters $f_1$ and $f_2$ are compatible if and only if there exists a tree $T \in UB(X)$ such that $f_1$ and $f_2$ are both compatible with $T$. 
\end{itemize}
\end{lemma}

\begin{lemma} \label{at3}
For a phylogenetic X-tree $T$, and a pair $\mathcal C = (f_1 , f_2 )$ of two-state characters on $X$ the following holds:
\begin{align*}
ps(\mathcal C , T) =
\begin{cases}
0, & \text{if $f_1 ,f_2$ are constant}; \\
1, & \text{if $f_1$ is compatible with $T$ and $f_2$ is constant (or vice versa)}; \\
2, & \text{if $f_1 ,f_2$ are compatible with $T$}; \\
\ge 3, & \text{otherwise.}
\end{cases}
\end{align*}
\end{lemma}

\begin{proof}
Let $T$ be a tree with two constant two-state characters $(f_1,f_2)$. Then for both characters the parsimony score is 0 and therefore the maximum parsimony score is 0. 
Next, without loss of generality, let $f_1$ be compatible with $T$ and $f_2$ constant, then
\begin{align*}
ps(\mathcal C ,T) = \underbrace{ps(f_1 ,T)}_{=1, ~ \text{by Lemma}~ \ref{tc}(a)} + \underbrace{ps(f_2 ,T)}_{=0, ~ \text{because $f_2$ is constant}} = 1 + 0 = 1.
\end{align*}
Now suppose that $f_1$ and $f_2$ are both not constant, but are compatible with $T$. By Lemma \ref{tc}(a) the parsimony score of each character is 1, so that $ps(\mathcal C ,T) = 2$. 
In all other cases we know that neither $f_1$ nor $f_2$ are constant, so $ps(f_1 ,T ) \neq 0 \neq ps(f_2 , T)$. Additionally, at least one of the two characters has a score of at least 2, because it is not compatible with $T$. Thereby the parsimony score is at least 3.
\end{proof}

\begin{lemma}\label{pars3}
For a pair $\mathcal{C} = (f_1,f_2)$ of two-state characters we have:
$$ \min_{T \in UB(n)} \{ ps(\mathcal{C},T)\}= 
\begin{cases}
0, & \text{if $f_1$ and $f_2$ are constant}; \\
1, & \text{if exactly one of $f_1$ or $f_2$ is constant}; \\
2, & \text{if neither $f_1$ nor $f_2$ are constant, but $f_1$ and $f_2$ are compatible}; \\
3, & \text{otherwise.}
\end{cases}
$$ 
\end{lemma}

\begin{proof} 
If $f_1$ and $f_2$ are both constant then the MP score of this pair of characters on any tree  is 0. 
Otherwise if exactly one of $f_1$ or $f_2$ is constant (without loss of generality, $f_1$ is constant) then $$\min_{T \in UB(n)} \{ps(\mathcal{C},T) \} = \min_{T \in UB(n)} \{ps(f_1,T) + ps(f_2,T) \} = \min_{T \in UB(n)} \{0 + ps(f_2,T) \}  = 1.$$
Now, suppose that  neither $f_1$ nor $f_2$ are constant, but $f_1$ and $f_2$ are compatible with each other. From Lemma \ref{tc}(b) there exists a tree $T \in UB(n)$ such that $f_1$ and $f_2$ are compatible with this tree, and so, Lemma \ref{at3} shows that $ps(\mathcal{C},T)=2$ and that $T$ is an MP tree for ${\mathcal C}$. 
In the last case $f_1$ and $f_2$ are not constant and $f_1$ and $f_2$ are incompatible with each other. For the corresponding $X$-splits, $A_1 \mid B_1$ and $A_2 \mid B_2$, we may suppose that $A_1$ and $A_2$ correspond to 0, $B_1$ and $B_2$ correspond to 1. 
Consider any phylogenetic $X$-tree of the type shown in Fig.~\ref{compatible} with the following leaf sets (none of which is empty); $X=A_1 \cap A_2, Y=A_1 \cap B_2, W=B_1 \cap A_2, Z=B_1 \cap B_2$. Note that there is no tree with a lower score by Lemma~\ref{at3}.

\bigskip

\begin{figure}[ht]
\centering
\includegraphics[scale=1]{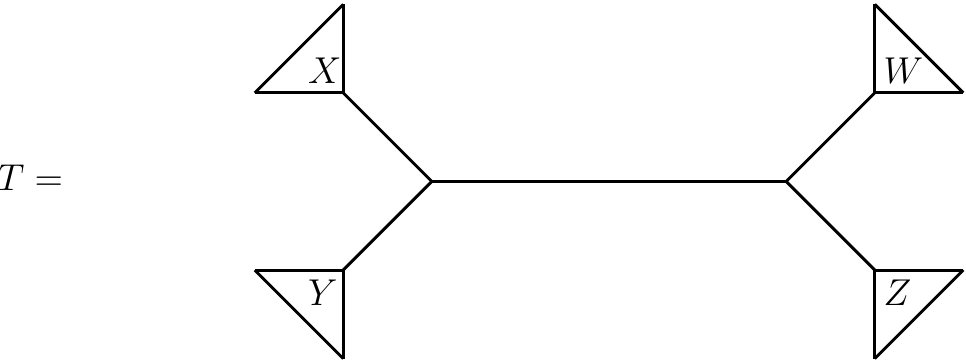}
\caption{$T \in UB(n)$ with four disjoint subtrees $X, Y, W$ and $Z$. The tree structure of $X, Y, W$ and $Z$ is unimportant.}
\label{compatible}
\end{figure}

\noindent
For $f_1$ all leaves in $X$ and $Y$ are in state 0 and all leaves in $W$ and $Z$ are in state 1. For $f_2$ all leaves in $X$ and $W$ are in state 0 and all leaves in $Y$ and $Z$ are in state 1. Then the MP score of the two characters on this tree is $\min_{T \in UB(n)} \{ps(f_1,T) + ps(f_2,T) \} = 1+2 = 3.$
\end{proof}

{\em Proof of Theorem 1:}

We describe an explicit counterexample for  $n=6, k=2$, namely the trees $T_1, T_2$ in  $UB(6)$ as shown in Fig. \ref{trees}, for which we will show that $mp_{2}(T_1) > mp_{2}(T_2)$. 

\begin{figure}[ht]
\centering
\includegraphics[scale=0.70]{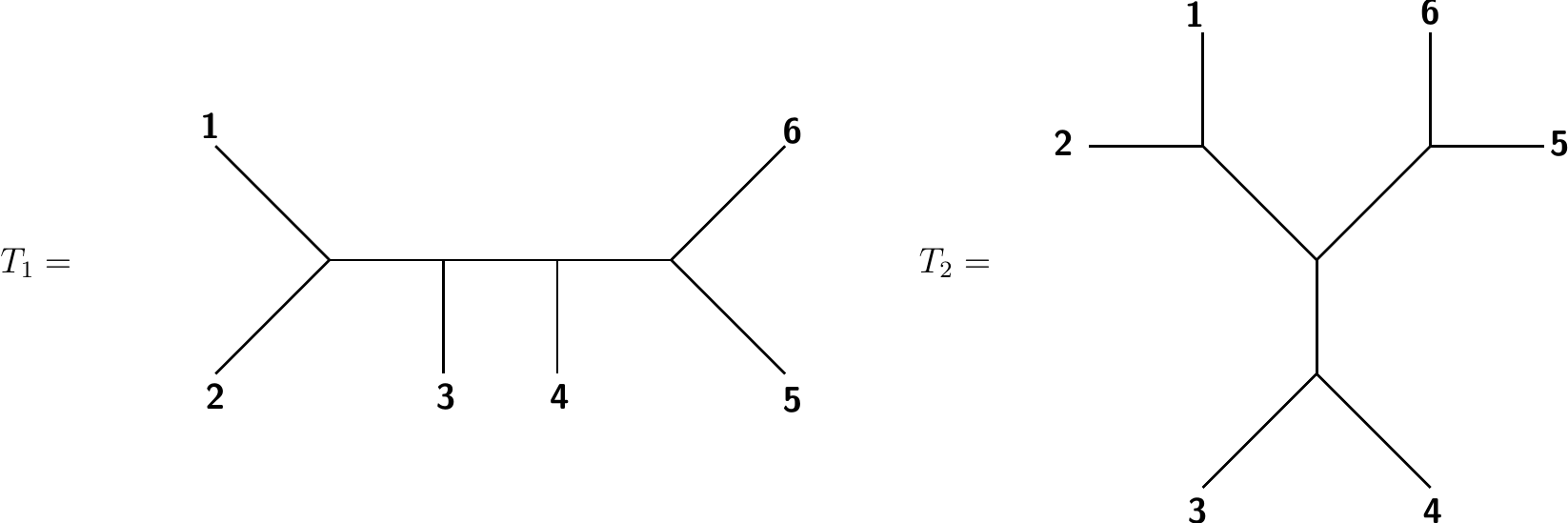}
\caption{$T_1, T_2 \in UB(6)$ with different tree shapes.}
\label{trees}
\end{figure}

Let $F := \{ f: X \rightarrow \{ 0,1 \} \}$ be the set of all two-state characters $f$ on $X = \{1, \dots, 6 \}$. Then $N := \{ f: X \rightarrow \{ 0,1 \} :   \# \text{ leaves in state $0$ is either $0, 1$, $n-1$ or $n$} \}$ is the set of all non-informative two-state characters $f$ on $X$. For each non-informative two-state character on a tree $T$, the character adds the same parsimony score to every tree (either 0 or 1). Furthermore, for any $T \in UB(n)$, define $I_j(T) := \{ f \in F \backslash N :  ps(f,T) = j \}$; $j = 1, 2, 3$.  Thus, $I_j(T)$ is the set of all informative two-state characters $f$ on $X$ which have a parsimony score $j$ on $T$, and when $n=6$,  $F$ is the (disjoint) union of the four sets $I_1(T), I_2(T), I_3(T), N$.   The number of characters in $I_1(T), I_2(T)$ and $I_3(T)$ is the same for any choice of $T \in UB(n)$ (this follows since the number of binary characters of parsimony score $k$ is the same for each choice of $T \in UB(n)$ (\cite{sem}, Theorem 5.6.2).  \\
Now we have a look at all possible cases to choose $f_1$ and $f_2$ from $N, I_1(T), I_2(T)$ and $I_3(T)$. 
The following statements about various exclusive cases apply for any $n$, but we will specialise soon to $n=6$ (since then the following seven cases exhaust every possibility). 

\noindent{\bf Case 1}: $f_1, f_2 \in N$. In this case, each tree $T \in UB(n)$ is an MP tree for this pair of characters, because there is no other tree with a lower score.
\\
{\bf Case 2:}  $f_1 \in N$ and $f_2 \in I_1(T)$ or ($f_1 \in I_1(T)$ and $f_2 \in N$).
\\
If the score of an informative character is 1, no tree achieves a better score than $T$, because only a non-informative character can have the score 0.  
Thus $T$ is an MP tree in Case 2.
\\
{\bf Case 3:}  $f_1 \in N$ and $f_2 \in I_2(T) \cup I_3(T)$ (or $f_1 \in I_2(T) \cup I_3(T), f_2 \in N$). 
\\
A non-informative $f_1$ contributes the same score to each tree $T \in UB(n)$. Moreover, when the  score of an informative character is 2,  this score can always be reduced.  Therefore $T$ with these characters is not an MP tree. 
Moreover, if  $f_1$ has the score 1 and $f_2 \in I_3(T)$, the score of the $T$ is $4$, and so, by Lemma \ref {pars3}, $T$ is not an MP tree. 
Finally,  if  $f_1$ has the score 0 and $f_2 \in I_3(T)$ one can always find a tree for $f_2$ which has a lower score. For this reason $T$ is never an MP tree.
\\
{\bf Case 4:}  $f_1, f_2 \in I_1(T)$. \\
As in Case 2 the scores of the characters cannot be improved by a tree different from $T$, so $T$ is an MP tree.
\\
{\bf Case 5:} For $j =2$ or $j=3$,   $f_1, f_2 \in I_j(T)$  \\
A tree $T$ with these characters has score 4 or score 6 and because of Lemma \ref {pars3}, $T$ is never an MP tree.
\\
{\bf Case 6:}   $f_1 \in I_3(T)$ and $f_2 \in I_1(T) \cup I_2(T)$ (or $f_1 \in I_1(T) \cup I_2(T)$ and $f_2 \in I_3(T)$). \\
A tree $T$ with these characters has score 4 or score 5 and by Lemma \ref {pars3}, $T$ is never an MP tree.
\\
{\bf Case 7:} $f_1 \in I_1(T)$ and $f_2 \in I_2(T)$ (or $f_1 \in I_2(T)$ and $f_2 \in I_1(T)$). \\
In this case,  we need to further investigate whether or not  $T$ is an MP tree.
\\ \\

When $n=6$ these represent all possible cases, and the only case where a different choice of $T \in UB(6)$ could affect whether or not $T$ is an MP tree is Case 7, which we consider in detail now.   
To  simplify the counting that follows, we may suppose, without loss of generality, that $f_1$ and $f_2$ both assign leaf 1 the state 0; moreover, for Case 7, we will just count 
the number of pairs of such characters $(f_1, f_2)$ where $f_1 \in I_1(T)$ and $f_2 \in I_2(T)$ for $T \in UB(6)$.

The character $f_1$ can be described by making a change on a single edge $\alpha$ of $T$, while for $f_2$ we require two changes, on edges labelled $\beta$ (we will see that these two $\beta$ edges are not always uniquely determined by $f_2$).  The placement of the two $\beta$ edges in relation to the $\alpha$ edges falls into three scenarios, referred to as (a), (b) and (c) in 
Fig.~\ref{fig3} (circles in this figure denote leaves or subtrees). 

\bigskip

\begin{figure}[ht]
\centering
\includegraphics[scale=0.50]{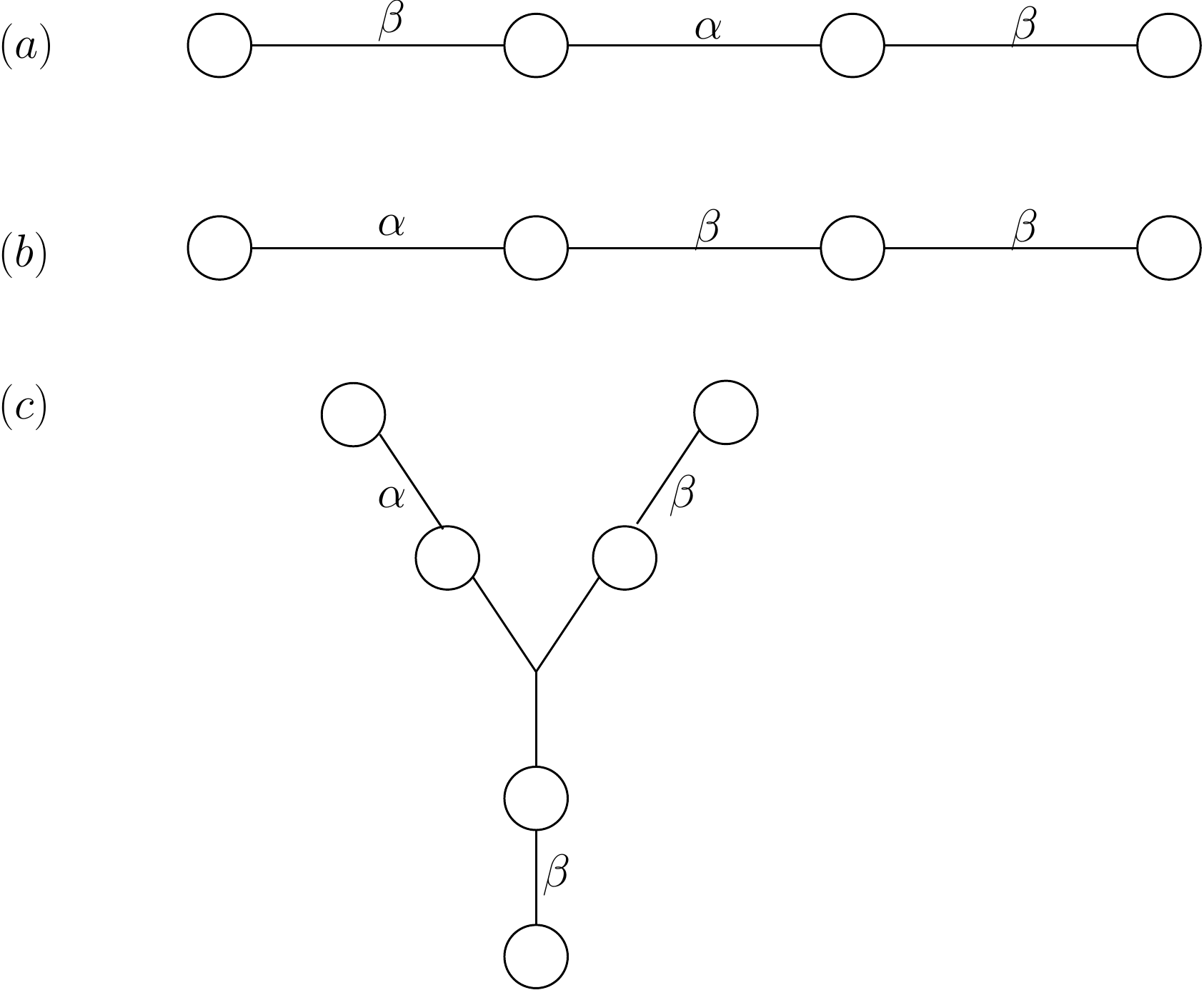}
\caption{$\alpha$: edge which corresponds to the single state change for $f_1 \in I_1(T)$ and $\beta$: two edges which correspond to the two state changes for $f_2 \in I_2(T)$.}
\label{fig3}.
\end{figure}

\noindent

Notice that in scenario (a) the splits induced by $f_1$ and $f_2$ are incompatible.  Thus, since the MP score of $T$ is 3, and this is best possible (by  
Lemma~\ref{pars3}, since the splits are incompatible), so $T$ is an MP tree under this scenario. 
In scenarios (b) and (c) the splits induced by $f_1$ and $f_2$ are compatible, and since the MP score of $T$ of 3 is not best possible (again 
by Lemma~\ref{pars3}, since the splits are compatible), $T$ is not an MP tree.  In summary,  $T$ is an MP tree if and only if scenario (a) applies.
We thus want to count the number of pairs of characters $(f_1, f_2)$ with $f_1 \in I_1(T)$ and $f_2 \in I_2(T)$ that correspond to scenario (a), and determine how this depends on the shape of the tree. 

For $T_1$ there are three possible edges for $\alpha$, so that $f_1 \in I_1(T)$ (see Fig.~\ref{cats}).
\begin{figure}[ht]
\centering
\includegraphics[scale=0.70]{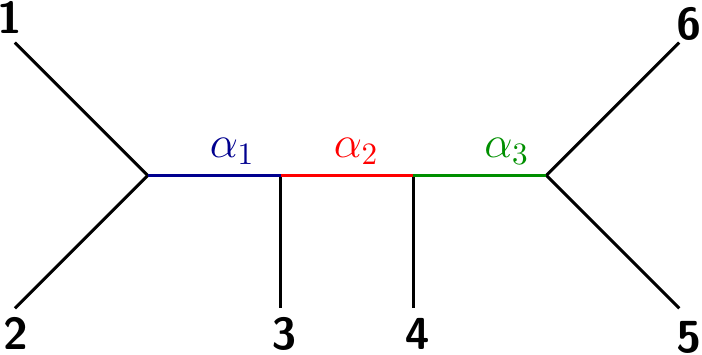}
\caption{$\alpha_1, \alpha_2, \alpha_3$ are the three possible edges for $T_1$, so that $f_1 \in I_1(T)$.}
\label{cats}
\end{figure}
\noindent
To arrive at scenario (a), the two $\beta$ edges must be on different sides of $\alpha$. So we have $2 \cdot 6=12$ different options to place the $\beta$ edges for each $\alpha_1$ and $\alpha_3$. For $\alpha_2$ we have $4 \cdot 4=16$ different options. But we are not only interested in how many places for changes we have. We  rather want to count the possible two-state characters $f_2$.  So we have to check if we count some two-state characters $f_2$ twice.  
We find that for every $\alpha_i ~(i=1,2,3)$ we count exactly two characters twice (see Fig.~\ref{cats2}). 
\begin{figure}[ht]
\centering
\includegraphics[scale=0.50]{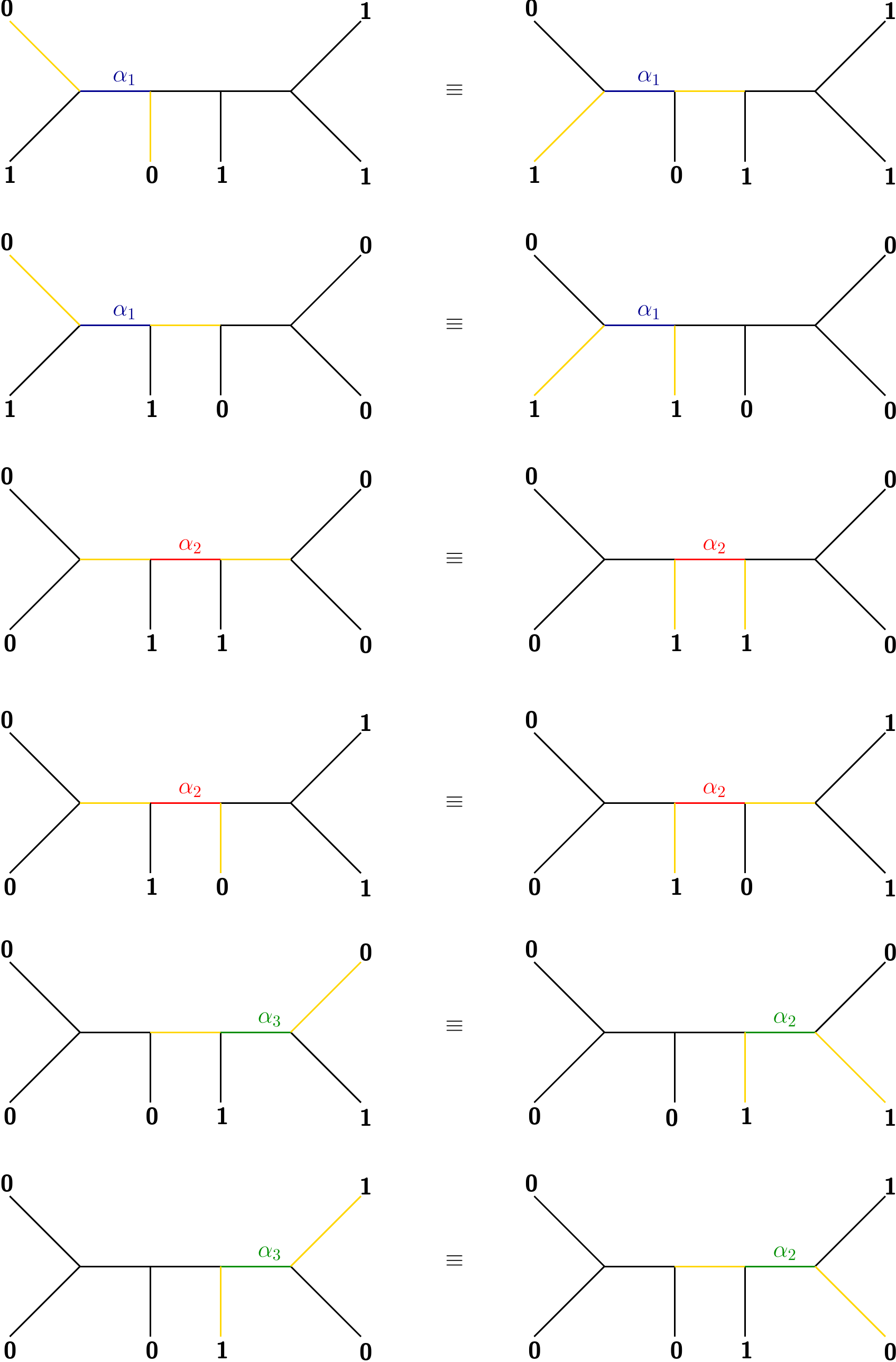}
\caption{For every $\alpha_i ~(i=1,2,3)$ we count two two-state characters twice.}
\label{cats2}
\end{figure} 
\noindent
So for $T_1$ in Case 6 we get $2 \cdot 6 + 2 \cdot 6 + 4 \cdot 4 - 6 = 34$ pairs of two-state characters $f_1$ and $f_2$ corresponding to scenario (a) (i.e. when $T_1$ is an MP tree). \\
Now we repeat this type of analysis for $T_2$, where we can also find three possible edges for $\alpha$; $\alpha_1, \alpha_2, \alpha_3$ (see Fig.~\ref{cen}).  But because of the symmetry of $T_2$ we just have to focus on one case. Here we focus on the edge $\alpha_1$; the  other two cases are analogous.
\begin{figure}[ht]
\centering
\includegraphics[scale=0.70]{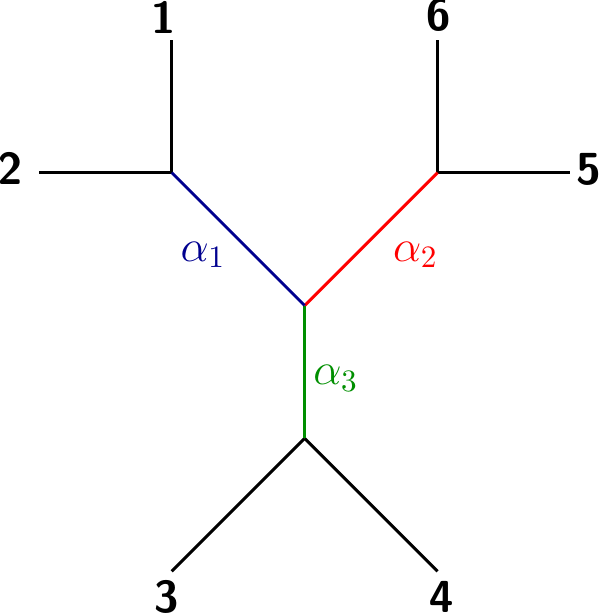}
\caption{$\alpha_1, \alpha_2, \alpha_3$ are the three possible edges for $T_2$, so that $f_1 \in I_1(T)$.}
\label{cen}
\end{figure}
\noindent
To arrive at scenario (a) we  have to place the $\beta$ edges on different sides of the $\alpha$ edge. So we get $3 \cdot (2 \cdot 6) = 36$ ways to achieve this for a combination of one $\alpha$ and the two $\beta$ edges. 
But once again  we must check which two-state characters are counted twice. Here there are two cases for every possible $\alpha$. The two cases for $\alpha_1$ are shown in Fig. \ref{snow2}.
\begin{figure}[ht]
\centering
\includegraphics[scale=0.50]{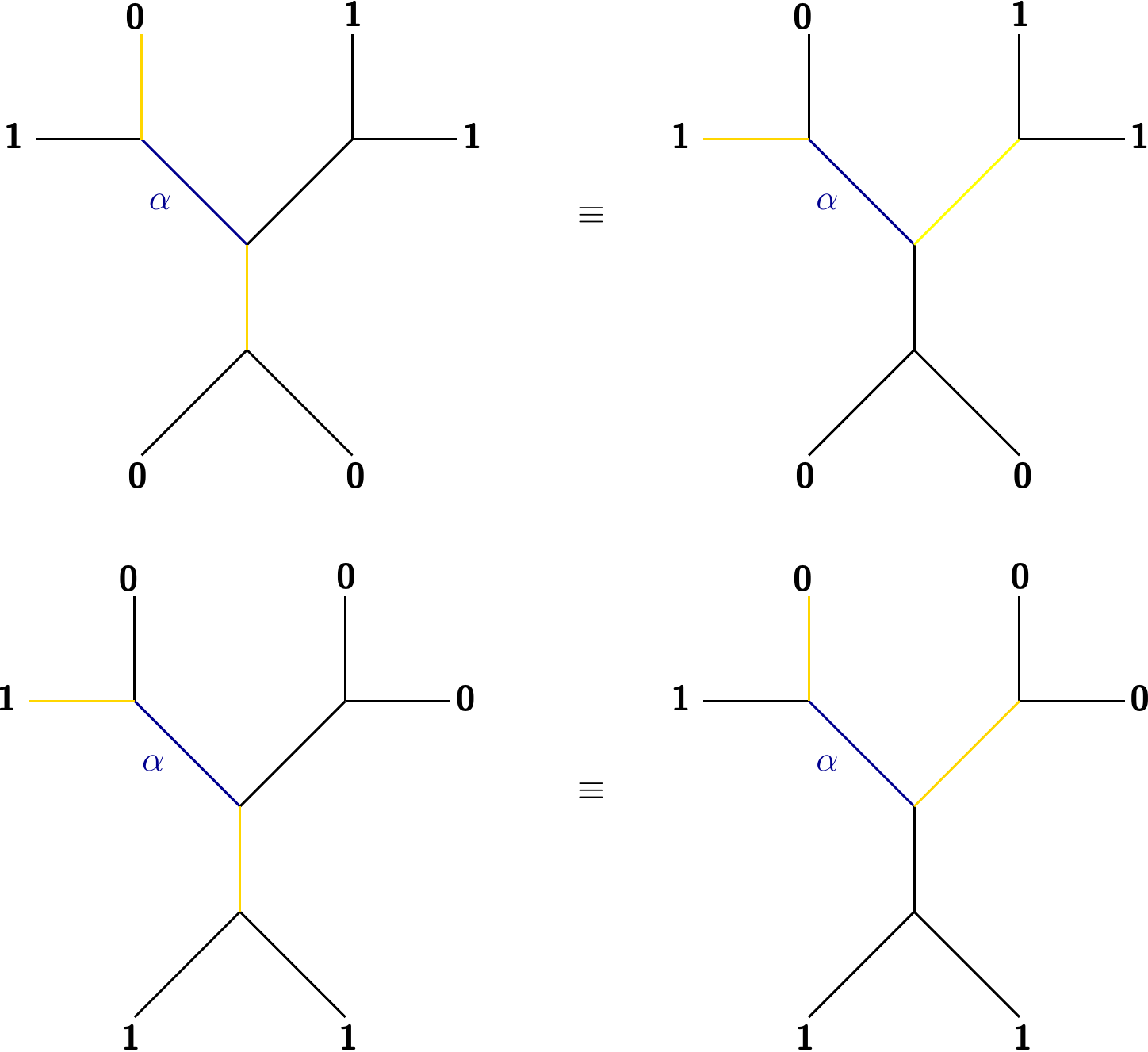}
\caption{For every $\alpha_i (i=1,2,3)$ we count two two-state characters twice. Here the two two-state characters we count twice for $\alpha_1$ are shown.}
\label{snow2}
\end{figure}
\noindent
Therefore we get $36 - 3 \cdot 2 = 30$ combinations of $f_1$ and $f_2$ so that $T_2$ is an MP tree. 
Now we see that for $T_1$ there are more combinations of $f_1 \in I_1$ and $f_2 \in I_2$ to be an MP tree than for $T_2$. For the reason that in all other cases the number of combinations of $f_1$ and $f_2$ to be an MP tree are the same, we can conclude that the probability that $T_1$ is an MP tree is higher than for $T_2$.
\hfill$\Box$

\end{document}